\documentclass{article} 
\usepackage{iclr2024_conference,times}
\usepackage{soul}

\usepackage{amsmath,amsfonts,bm}

\def\eqref#1{equation~\ref{#1}}

\def\1{\bm{1}}

\DeclareMathAlphabet{\mathsfit}{\encodingdefault}{\sfdefault}{m}{sl}
\SetMathAlphabet{\mathsfit}{bold}{\encodingdefault}{\sfdefault}{bx}{n}

\newcommand{\tuple}[1]{\ensuremath{\langle#1\rangle}}
\newcommand{\set}[1]{\ensuremath{\{#1\}}}
\newcommand{\powerset}[1]{\ensuremath{\mathbb{P}(#1)}\xspace}
\usepackage{macros}
\usepackage{amsmath}
\usepackage{peanutgallery}

\newcommand{\revise}[1]{#1}
\newcommand{\remove}[1]{}

\title{Lemur:  Integrating Large Language Models in Automated Program Verification}

\author{%
  Haoze Wu\thanks{This work was mostly done during an internship at VMware Research.}, \  Clark Barrett\\
  Department of Computer Science\\
  Stanford University\\
  \texttt{\{haozewu, barrett\}@stanford.edu} \\
  \And
  Nina Narodytska \\
  VMware Research \\
  VMware by Broadcom\\
  \texttt{n.narodytska@gmail.com} \\
}

\newif\ifextended

\extendedtrue
\ifextended
\newcommand{\extref}[1]{#1}
\else
\newcommand{\extref}[1]{the extended version of the paper~\citep{wu2023lemur}}
\fi

\ifextended
\newcommand{\ext}[1]{#1}
\else
\newcommand{\ext}[1]{}
\fi

\iclrfinalcopy 
\begin{document}

\maketitle

\begin{abstract}
The demonstrated code-understanding capability of LLMs raises the question 
of whether they can be used for automated program verification, a task that 
demands high-level abstract reasoning about program properties that is challenging for verification tools. We propose a general methodology to combine the power of LLMs and automated reasoners for automated program verification. We formally describe this methodology as a set of transition rules and prove its soundness. We instantiate the calculus as a sound automated verification procedure and demonstrate practical improvements on a set of synthetic and competition benchmarks.

%


\end{abstract}

\section{Introduction}

AI-powered language models are being routinely used to help developers. 
Examples include program synthesis from natural language descriptions by GPT-4
~\citep{openai2023gpt4} or Github Copilot~\citep{copilot,githubcopilot}, 
and repairing code~\citep{deeprepair}, among others. These models have shown impressive results in generating correct code in many programming languages. %

An important research question is whether modern AI models are capable of understanding the logic behind the programs they analyze. Recently, several approaches have been proposed to combine the strengths of formal verification and Large Language Models (LLMs) that demonstrate such capabilities. For example, \cite{pmlr-v202-pei23a} made an important step in this direction by investigating whether LLMs can generate program properties, namely, program invariants, which remains a crucial and challenging task for automated program verification~\citep{clarke2018handbook}. The authors demonstrated that LLMs are effective in generating program invariants on a set of synthetic Java programs. Another example is the recent work by \cite{charalambous2023new}, who demonstrated that LLM models can be used to repair vulnerabilities in code, given examples of incorrect behavior. They provided compelling evidence of the complementary strengths of LLMs, which serve as a generator for code repair snippets, and formal techniques, which are used to check the correctness of the generated code. While these approaches show promise in program analysis tasks, they do not provide a formalization of the interaction between LLMs and formal verifiers, they require manual effort, or they are limited to the invariant generation process as a stand-alone procedure.

In this work, we propose a novel LLM-powered framework, \sys, for automated program verification tasks.
Our key idea is to combine LLMs' ability to perform abstract high-level reasoning and automated reasoners' ability to perform precise low-level reasoning. Specifically, LLMs are employed to propose program invariants in the form of sub-goals, which are then checked by automated reasoners. This transforms the program verification tasks into a series of deductive steps suggested by LLMs and subsequently validated by automated reasoners. Our main contributions are:
\vspace{-0.5pc}
\begin{itemize}[leftmargin=*]\itemsep0.08em 
\item a novel framework for  combining LLMs  and automated reasoners for  program verification; 
\item a presentation of \sys as a proof system and a proof of its soundness, which to the best of our knowledge, is the first formalization of 
such a hybrid approach;
\item an instantiation of the \sys calculus that gives a sound and terminating algorithm;
\item 
an implementation of the proposed framework (using OpenAI's GPT models); 
\item an experimental evaluation of \sys on two sets of benchmarks that demonstrates its efficiency compared with both existing AI-powered and conventional verification tools.
\end{itemize}
\vspace{-0.5pc}
\revise{ We highlight that \sys is the first \emph{fully automated} framework combining   LLMs and  reasoners.}

\section{Definitions}
\label{sec:prelim}

Given a program $\program : \Program$, a \emph{reachability property}, or simply \emph{property}, is a tuple $\property\is\tuple{\predicate, \linenumber}$, where $\predicate : \Predicate$ is a Boolean \emph{predicate} of the program state and $\linenumber : \Nat$ is a \emph{program line}. 
The \emph{negation} of \property, denoted $\neg\property$, is defined as \tuple{\neg\predicate, \linenumber}. 
Next we introduce several useful definitions and their properties.
\begin{definition}
A property $\property=\tuple{\predicate, \linenumber}$ is an \emph{invariant} on $\program$, denoted $\invariant(\program,p)$, iff $\property$ holds (i.e., \predicate always evaluates to true at line \linenumber) for all possible executions of the program \program.
\end{definition}
\begin{example}\label{exm:intro}
Consider the simple program $\program$ in Figure~\ref{fig:running-example} (the first frame, top row).   $\program$ instantiates an unsigned 32-bit integer variable \texttt{x} to 0 and increases its value by 4 on each loop iteration. A property $\property=\tuple{\predicate, \linenumber}$ is specified on the $4$th line, so $\predicate = (\texttt{x}~!\!\!=30)$  and $\linenumber=4$ for this property. It is easy to see that $\property$ is an invariant as   \texttt{x} will never be divisible by 3, for example. $\blacksquare$
\end{example}
\vspace{-0.2cm}

An assumption $\propertyAlt=\tuple{\predicate, \linenumber}$ is a condition that is assumed in a program.
\begin{definition}
An \emph{assumption} $\propertyAlt=\tuple{\predicate, \linenumber}$ modifies a program \program as follows:
\begin{enumerate}[nolistsep]
    \item if $\predicate$ holds at line  \linenumber during some execution of \program, then \program continues execution without changes;  
    \item if $\predicate$ does not hold at line  \linenumber during some execution of \program, then that execution terminates at $\linenumber$.
\end{enumerate}
\end{definition}
\vspace{-0.2cm}

We use $\programAlt= \asm(\program,q)$ to denote the modification of $\program$  with the assumption $\propertyAlt$.  An assumption can itself be an invariant. We now introduce a special notion of an \emph{implication}.


\begin{definition}
Let $\program$ be a program, and \property, \propertyAlt be properties on \program. We say that $\propertyAlt$ \emph{implies} $\property$ with respect to $\program$, denoted \imply{\propertyAlt}{\program}{\property}, iff $\property$ is an invariant on $\asm(\program, \propertyAlt)$.
\end{definition}
\vspace{-0.3cm}

\begin{example}
Consider the program $\program$ in Figure~\ref{fig:running-example} and an assumption $\propertyAlt=\tuple{\predicate=(\texttt{x\%4==1)}, \linenumber=3}$, with $\programAlt= \asm(\program,q)$
(depicted in the first frame in the bottom row of Figure~\ref{fig:running-example}).
To see the difference between $\program$ and $\programAlt$, observe that the loop is executed only once in  $\programAlt$.  This is because
\texttt{x=0} when entering the loop, so \texttt{(x\%4)!\!=1}, which is the formula $\predicate$ in $q$, does not hold, and therefore, \programAlt terminates. 
If we consider an alternative assumption $\propertyAlt'=\tuple{\predicate=(\texttt{x\%4==0)}, \linenumber=3}$
(depicted in the fourth frame at the bottom of Figure~\ref{fig:running-example}),
we can see that its predicate \predicate holds for all executions. Hence, $\propertyAlt'$ is an invariant for $\program$. 
Finally, we can see $\propertyAlt'\xrightarrow[\program]{} \property$, where \property is from Example~\ref{exm:intro}.  
$\blacksquare$
\end{example}
\vspace{-0.3cm}
The following propositions follow from the definitions above. 
\begin{proposition}
Let $\program$ be a program, and \property, \propertyAlt be properties on \program:
\begin{itemize}[nolistsep]
\item  The property \property is an invariant on $\program$ if 
\propertyAlt is an invariant on \program  and
\propertyAlt implies \property with respect to \program. More formally,  $(\invariant(\program,\propertyAlt) \wedge \imply{\propertyAlt}{\program}{\property}) \Rightarrow \invariant(\program,\property)$.
\item
The property \property is not an invariant on \program if 
 the property \property is not an invariant on $\programAlt\is\asm(\program,q)$. More formally, \revise{$\neg \invariant(\programAlt,\property)   \Rightarrow \neg \invariant(\program,\property)$}
\remove{$\neg \invariant(\program,\propertyAlt)   \Rightarrow \neg \invariant(\programAlt,\property)$}.
\end{itemize}
\label{prop:imply}
\end{proposition}

\begin{proposition}
For any property $\property$ on a program \program, $\property\xrightarrow[\program]{} \property$.
\label{prop:self-imply}
\end{proposition}

\begin{proposition}
For any properties $\property, \propertyAlt, \propertyAltAlt$ on a program \program, if $\property\xrightarrow[\program]{} \propertyAlt$ and $\propertyAlt\xrightarrow[\program]{} \propertyAltAlt$, then $\property\xrightarrow[\program]{} \propertyAltAlt$.
\label{prop:chain-imply}
\end{proposition}
\vspace{-0.3cm}
Note that it is possible that neither a property \property nor its negation $\neg\property$ is an invariant on a program. 
\begin{example}
Consider again our example from Example~\ref{exm:intro} and two properties  at line $3$: 
$\property=\tuple{\predicate=(\texttt{x\%8==4}), \linenumber=3}$
and  $\property'=\tuple{\predicate'=(\texttt{x\%8!=4}), \linenumber=3}$. Neither $\property$ nor $\property'$ is an invariant on $\program$. 
On the first loop iteration, we have that $\texttt{x=0}$ before line $3$, so $\predicate'$ holds and $\predicate$ does not at line $3$.
But on the second loop iteration, we have that $\texttt{x=4}$ before line $3$,  so $\predicate$ holds and $\predicate'$ does not.
$\blacksquare$    
\end{example}

\begin{definition}\label{def:stable}
A property $\property=\tuple{\predicate, \linenumber}$ is \emph{stable} for $\program$, denoted $\stable(\program,\property)$, if, for each execution of the program, either \predicate always evaluates to true at line \linenumber or \predicate always evaluates to false at line \linenumber.
\end{definition}
\begin{example}
\revise{Consider an example to illustrate the definition of stability.}
\begin{verbatim}
Line 1: uint32_t x=rand();
Line 2: assert(x==1);
\end{verbatim} 
\revise{The property in Line 2 is stable, as it always evaluates to true or false within a single execution.$\blacksquare$}

\end{example}

An invariant must be stable, but a property that is not an invariant might still be stable. For example, any property on a program without loops is stable. If $\property$ is stable, then $\neg\property$ is also stable.  Lemma~\ref{lemma:stable} exploits stable 
invariants (see a proof in \extref{App.~\ref{app:sec:def}}).

\begin{lemma}
Consider a program \program, two properties \property, \propertyAlt on \program, and a program $\programAlt\is\asm(\program,\propertyAlt)$.
The property \property is an invariant on $\program$ if: (i) \propertyAlt is stable for $\program$; (ii) \propertyAlt implies \property with respect to $\program$; and (iii) $\neg\propertyAlt$ implies \property with respect to $\program$. More formally:
$\stable(\program,\propertyAlt) \wedge (\imply{\propertyAlt}{\program}{\property}) \wedge (\imply{\neg\propertyAlt}{\program}{\property})\Rightarrow \invariant(\program, \property).$
\label{lemma:stable}
\end{lemma}
\vspace{-0.5cm}

Assume we have a verifier $\verifier : \Program \times \powerset{\Property} \times \Property \mapsto \set{\true, \false, \unknown}$,
which takes as inputs a program \program, a set of assumptions \assumptions and a property \property, and checks whether \assumptions implies \property. 
More precisely, given set of assumptions  $\assumptions = \{\propertyAlt_1,\ldots,\propertyAlt_n\}$, we construct 
a new program $\programAlt = \asm(\asm((\ldots, \asm(\program,\propertyAlt_1)),\propertyAlt_{n-1}),\propertyAlt_n)$, and the verifier checks if $p$ is an invariant on this program. Hence, a statement that \assumptions implies \property on $\program$ means that \property is an invariant on \programAlt.
%
%
We further assume that \verifier is \emph{sound}, meaning if \verifier returns \true, then \assumptions implies \property, and if \verifier returns \false, then \assumptions does not imply \property. Note that \assumptions can be empty, in which case the verifier just checks whether \property is an invariant. When the verifier \verifier returns \true, we say \property is proven; and when \verifier returns \false, we say the property is \emph{falsified}. 
\verifier is incomplete, meaning that \verifier can return \unknown.

In practice, \verifier can be instantiated with an automated program verifier such as \cbmc~\citep{cbmc}, \esbmc~\citep{esbmc}, or \uautomizer~\citep{uautomizer}. We provide an overview of the main techniques employed by these tools in \extref{App.~\ref{app:background}} and note here that a crucial challenge shared across existing verifiers is the automatic decomposition of a verification task into smaller, more manageable sub-tasks. This decomposition requires high-level reasoning that is difficult to automate using conventional formal methods, but plausibly could be performed by LLMs, with their documented code-understanding capability. However, it is crucial to preserve soundness when LLMs are used to automatically perform this high-level reasoning in program verification tasks.

\section{\sys: Integrating LLMs in Program Verification}
We present \sys, a proof calculus that combines LLMs and automated reasoners to prove properties of programs. The calculus operates over a \emph{configuration}, which is either one of the distinguished symbols \set{\success, \fail} or a tuple \tuple{\program, \assumptions, \trail}, where \program is a program, \assumptions is either $\varnothing$ or a singleton representing the assumption, and \trail is a list of properties referred to as \emph{proof goals}.  \trail itself is referred to as a \emph{trail}. The last element of \trail represents the current property to prove. 
\revise{
We use the notation $\concat$ 
to denote a concatenation of two lists. In particular, $\trail = \trail' \concat \property$ means that $\trail$ is a concatenation of a trail $\trail'$
and a property $\property$, where $\property$ is the last element of $\trail$.
}
The rules describe the conditions under which a certain configuration can transform into another configuration. In this calculus, deciding whether $\invariant(\program, \property)$ holds reduces to finding a sequence of valid rule applications from the \emph{starting configuration} \tuple{\program, \varnothing, [\property]} to either \success or \fail. 

Our calculus performs oracle calls to LLMs to propose new properties and revise them.
The  oracle $\prompterPropose: \Program  \times \Property \mapsto \powerset{\Property}$  
proposes new properties, given a program and the current proof goal as inputs. 
A key hypothesis here is that LLMs are capable of generating new properties that are likely to (i) be invariants, and (ii) imply the proof goal given in a prompt. We will discuss strategies to generate prompts in Section~\ref{sec:strategy}. Importantly, properties generated by an LLM are treated as assumptions until we can prove that they are invariants of the original program.
The oracle $\prompterRepair$  revises previously proposed properties, e.g.,
if we determine that a property \propertyAlt previously produced by $\prompterPropose$ does not hold or does not imply the current proof goal. In this case, we request an LLM to repair 
\propertyAlt. We have $\prompterRepair: \Program \times \Property \times \Property \times \set{\false,\unknown} \mapsto \powerset{\Property}$,
whose inputs comprise a program, two properties, and a solver return value.
The first property is the current proof goal, and the second property \propertyAlt is an assumption previously proposed by oracles.
The output of \prompterRepair is a new set of properties. In practice, we implement it with a prompt to an LLM to either correct or strengthen \propertyAlt (see Section~\ref{sec:strategy}).
Finally, the calculus performs an external call to a verifier \verifier to check whether a property holds.


The proof rules of \sys are shown in Fig.~\ref{fig:rules}. 
Each rule defines a set of conditions that must hold for the rule to be applicable. Note 
again that these conditions permit invocations of LLMs and/or verifiers. The rules within the calculus can be partitioned into four groups.

\begin{figure}[t!]
\small
\centering
\begin{prooftree}
\AxiomC{$\trail\is \trail'\concat\property \quad \verifier(\program, \assumptions,\,\property) = \unknown \quad \propertyAlt \in \prompterPropose(\program, \, \property)$ }
\RightLabel{(\rulepropose)}
\UnaryInfC{$\program, \assumptions, \trail \transitionInto \program, \set{\propertyAlt}, \trail$}
\end{prooftree}
\begin{prooftree}
\AxiomC{$ 
\assumptions\is\set{\propertyAlt}
\quad
\trail\is\trail'\concat\property \quad \verifier(\program, \assumptions,\,\property) = \true$ }
\RightLabel{(\ruledecide)}
\UnaryInfC{$\program, \assumptions, \trail \transitionInto \program, \varnothing, \trail\concat\propertyAlt$}
\end{prooftree}
\begin{prooftree}
\AxiomC{$
\trail\is\trail'\concat\property\concat\propertyAlt \quad 
\verifier(\program, \assumptions, \,\propertyAlt) \neq \true
\quad \propertyAlt' \in \prompterPropose(\program, \, \property)$ }
\RightLabel{(\rulebacktrack)}
\UnaryInfC{$\program, \assumptions, \trail \transitionInto \program, \set{\propertyAlt'},\, \trail'\concat\property$}
\end{prooftree}
\begin{prooftree}
\AxiomC{$
\assumptions\is\set{\propertyAlt}
\quad
\trail\is\trail'\concat\property \quad \verifier(\program, \assumptions,\,\property) = \unknown \quad
\propertyAlt' \in \prompterRepair(\program,\, \property, \propertyAlt, \unknown)$ }
\RightLabel{(\rulerepairone)}
\UnaryInfC{$\program, \assumptions, \trail\transitionInto \program, \set{\propertyAlt'},\, \trail'\concat\property$}
\end{prooftree}
\begin{prooftree}
\AxiomC{$
\assumptions\is\varnothing
\quad
\trail\is\trail'\concat\property\concat\propertyAlt 
\quad
\verifier(\program, \assumptions,\,\propertyAlt) = \false \quad 
\propertyAlt' \in \prompterRepair(\program,\, \,\property, \propertyAlt,\, \false)$ }
\RightLabel{(\rulerepairtwo)}
\UnaryInfC{$\program, \assumptions, \trail\transitionInto \program, \set{\propertyAlt'},\, \trail'\concat\property$}
\end{prooftree}
\begin{prooftree}
\AxiomC{$\assumptions\is\varnothing\quad \trail\is\trail'\concat\property \quad \verifier(\program, \assumptions,\,\property) = \true$ }
\RightLabel{(\rulesuccessone)}
\UnaryInfC{$\program, \assumptions, \trail \transitionInto \success$}
\end{prooftree}
\begin{prooftree}
\AxiomC{$
\assumptions\is\varnothing
\quad
\trail\is\trail'\concat\property\concat\propertyAlt
\quad \stable(\program,\propertyAlt) \quad \verifier(\program, \set{\neg\propertyAlt},\,\property)=\true$ }
\RightLabel{(\rulesuccesstwo)}
\UnaryInfC{$\program, \assumptions, \trail\transitionInto \success$}
\end{prooftree}
\begin{prooftree}
\AxiomC{$
\trail\is [\property]
\quad
\verifier(\program, \assumptions,\,\property) = \false$ }
\RightLabel{(\rulefail)}
\UnaryInfC{$\program, \assumptions, \trail \transitionInto \fail$}
\end{prooftree}
\vspace{-2mm}
\caption{Deductive rules of the \sys calculus.}
\label{fig:rules}
\vspace{-0.5cm}
\end{figure}

The first group contains rules that are responsible for generating new proof goals for specific configurations. These rules are 
\rulepropose, \rulerepairone, and \rulerepairtwo.
The \rulepropose rule states that if the verifier is unable to prove or disprove the current proof goal \property, we can invoke the \prompterPropose oracle to obtain a property \propertyAlt and update \assumptions to be \set{\propertyAlt}. 
The \rulerepairone rule can be applied when the current assumption \propertyAlt is not sufficient for the verifier to prove the current proof goal \property. In this case, we can use the oracle \prompterRepair to propose ways to strengthen \propertyAlt and choose one of them, \propertyAlt', as the new assumption. The \rulerepairtwo rule can also be applied when \propertyAlt is already in the trail but is falsified by the verifier \verifier. In this case, we can use \prompterRepair to repair \propertyAlt and update \assumptions accordingly.

The second group
contains only the \ruledecide rule.  This rule allows
an assumption \propertyAlt to become a new goal (i.e., be appended to \trail) if the verifier \verifier is able to prove that \propertyAlt implies the current goal.

The third group allows \sys to recover from faulty assumptions. The \rulebacktrack rule is applicable when there are at least two elements in the trail and the verifier cannot prove the current proof goal.  It allows \sys to revert to the previous proof goal (the second to last property in the trail \trail) and pick a different assumption suggested by \prompterPropose.  Note that \rulebacktrack might not be the only applicable rule. For example, \rulerepairone or \rulerepairtwo could also be applicable. In practice, we use a strategy to decide between multiple applicable rules. This is discussed in Sec.~\ref{sec:strategy}. 

The final group specifies three termination conditions for the calculus.
The \rulesuccessone rule states that whenever the assumption is empty and the verifier is able to prove the current proof goal (i.e., the last property \property in the trail \trail), we can transition into the \success state. Note that if the verifier can prove the original property, then this rule can be directly applied to the starting configuration to reach \success.
%
\rulesuccesstwo states that if the last two elements of the trail \trail are \property and \propertyAlt, the current proof goal $\propertyAlt := \tuple{\phi, \linenumber}$ is stable (as defined in Sec.~\ref{sec:prelim}), and the verifier is also able to also prove \property under the assumption\tuple{\neg\phi, \linenumber}, then \property is an invariant (as a consequence of Lemma~\ref{lemma:stable}), and we can transition to \success. The \rulesuccesstwo rule constitutes a way to utilize an \emph{incorrect sub-goal} \propertyAlt proposed by the LLM-based oracles to decompose the verification task: we separately reason about the cases when \propertyAlt holds and when it does not hold. 
Finally, if the verifier \verifier proves that the original property is not an invariant, whether under an assumption or not, then we transition to the \fail state using \rulefail.

Note that the program $\program$ remains unchanged in every rule. We keep it as part of the state for two reasons. First, it is convenient for keeping the calculus self-contained, as \program is an input to the verifiers and the oracles. Second, in the future, it might be possible to augment the calculus with rules that update $\program$, by, for example, rewriting the program using LLMs in an invariant-preserving manner. 

We state the following soundness properties about \sys. The proof is presented in \extref{App.~\ref{app:subsec:soundness}}.

\begin{theorem}[Soundness]
Given a property $\property_0$ and a program \program, if \success is reached by a sequence of valid rule applications starting from \tuple{\program, \varnothing, [\property_0]}, then $\property$ is an invariant on \program.
\label{thm:soundness}
\end{theorem}

\begin{theorem}[Soundness 2]
Given a property $\property_0$ and a program \program, if \fail is reached by a sequence of valid rule applications starting from \tuple{\program, \varnothing, [\property_0]}, then $\property$ is not an invariant on \program.
\label{thm:soundness2}
\end{theorem}
\begin{figure}[t]
\centering
\resizebox{\textwidth}{!}{  
 \begin{tikzpicture}[node distance=2pt,
        box/.style={align=left,inner sep=1pt, fill=gray!5},
        marrow/.style={single arrow, 
              single arrow head extend=1mm,
              execute at begin node={\strut}}]
\scriptsize
\node (A) [box, draw] {
\begin{lstlisting}
uint32_t x=0;
while (rand()){
x+=4;
assert(x!=30);
}
\end{lstlisting}    
};

\node [above = 1pt of A, draw=none]{\verifiercolor{\verifier: \unknown}};

\node (B) [marrow, right=of A, draw]{
\rulepropose
};
\node [above = 1pt of B, draw=none, align=center]{\llmcolor{\prompterPropose}};

\node (C) [box, right=of B, draw] {
\begin{lstlisting}
uint32_t x=0;
while (rand()){
assume(x%2==0);
x+=4;
assert(x!=30);
}
\end{lstlisting}    
};
\node [above = 1pt of C, draw=none]{\verifiercolor{\verifier: \unknown}};

\node (D) [marrow, right=of C, draw]{\textbf{Rep. 1}};
\node [above = 1pt of D, draw=none, align=center]{\llmcolor{\prompterRepair}};

\node (E) [box, right=of D, draw]{
\begin{lstlisting}
uint32_t x=0;
while (rand()){
assume(x%4==0);
x+=4;
assert(x!=30);
}
\end{lstlisting} 
};
\node [above = 1pt of E, draw=none]{\verifiercolor{\verifier: \true}};

\node (F) [marrow, right=of E, draw]{\ruledecide};
\node (G) [box, right=of F, draw]{
\begin{lstlisting}
uint32_t x=0;
while (rand()){
assert(x%4==0);
x+=4;
}
\end{lstlisting} 
};

\node [above = 1pt of G, draw=none]{\verifiercolor{\verifier: \true}};

\node (H) [marrow, right=of G, draw]{\textbf{Succ. 1}};

\node (I) [box, below=20pt of C, draw]{
\begin{lstlisting}
uint32_t x=0;
while (rand()){
assume(x%4==1);
x+=4;
assert(x!=30);
}
\end{lstlisting} 
};

\node [above = 1pt of I, draw=none]{\verifiercolor{\verifier: \true}};

\node (J) [marrow, right=of I, draw]{\ruledecide};
\node (K) [box, right=of J, draw]{
\begin{lstlisting}
uint32_t x=0;
while (rand()){
assert(x%4==1);
x+=4;
}
\end{lstlisting} 
};

\node [above = 1pt of K, draw=none]{\verifiercolor{\verifier: \false}};

\node (L) [marrow, right=of K, draw, align=center]{\textbf{Backtrack}};
\node [above = 1pt of L, draw=none]{\llmcolor{\prompterPropose}};

\node (M) [box, right=of L, draw]{
\begin{lstlisting}
uint32_t x=0;
while (rand()){
assume(x%4==0);
x+=4;
assert(x!=30);
}
\end{lstlisting} 
};
\node [above = 1pt of M, draw=none]{\verifiercolor{\verifier: \true}};

\node (N) [marrow, right=of M, draw]{...};

\node (O) [marrow, draw, rotate=-45] at (1.9, -1.1) {\rulepropose};
\node [above = 1pt of O, draw=none, align=center, rotate = -45]{\llmcolor{\prompterPropose}};

\tikzstyle{roundedbox} = [draw=none, text width=4cm, align=center, inner sep=2pt]

\node (M) [roundedbox] at (0.7,-2.4) {
\vspace{-2mm}
\begin{verbatim}
...
x+=4;
assert(x!=30);
...
List invariants that prove 
the assertion. Your answer 
should look like
assert(...) //Line number
------------------------
assert(x % 2 == 0); //Line 2
assert(x % 4 == 1); //Line 2
\end{verbatim} 
};

\coordinate (start) at (0.8,-1.8);
\coordinate (end1) at (1.4,-1.2);
\coordinate (end2) at (1.5,-0.3);

\draw (start) -- (end1);
\draw (start) -- (end2);
\end{tikzpicture}
}
\vspace{-3mm}
\caption{A running example of executing the \sys calculus.}
\label{fig:running-example}
\vspace{-4mm}
\end{figure}
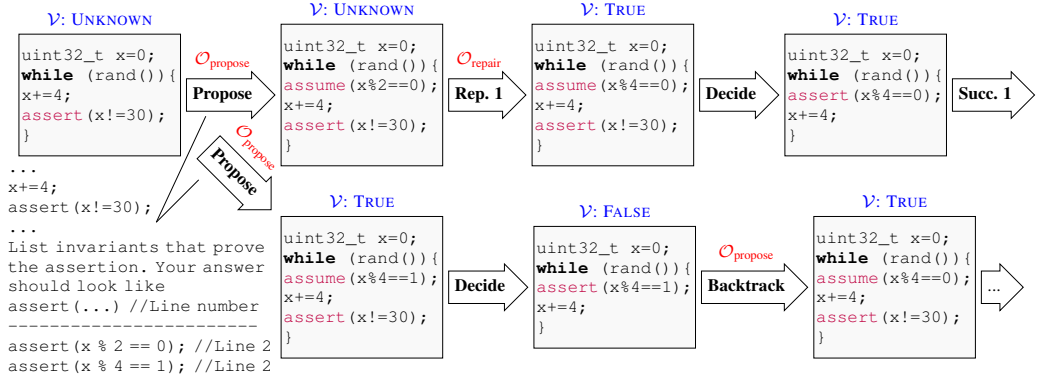
\begin{example}
To provide more intuition about the proof system and to motivate the design choices when instantiating \sys, we consider again our running example. Figure~\ref{fig:running-example} illustrates how \sys can be used to verify properties in practice. In Figure~\ref{fig:running-example} each frame represents a state of the program. Transitions between states are depicted by arrows, with each arrow marked with the rule applied to execute this transition. 
In this example, our goal is to prove the property \texttt{x!\!=30} in a while loop that keeps adding 4 to an unsigned 32-bit integer variable \texttt{x}. We note that this particular verification task is adapted from a similar one in the SV-COMP competition.\footnote{\url{https://sv-comp.sosy-lab.org/2023/results/results-verified/META_ReachSafety.table.html\#/table?filter=id_any(value(jain_5-2))}}
While seemingly trivial, during the competition, 19 out of the 24 participating tools (including the overall winner of the competition \uautomizer) were not able to solve this benchmark. 

The initial configuration is \tuple{\program, \varnothing, [\property])}, where \program is the given program and $\property\is\tuple{\texttt{x!\!=30}, 3}$.\footnote{3 is the line number (in the snippet) where the predicate is asserted.} 
Suppose the verifier \verifier is unable to solve this problem and returns \unknown. In this case, we need to generate a new proof goal, so the only rule we can apply is \rulepropose. To do so, we invoke the LLM-based oracle \prompterPropose to obtain a set of new properties that are potentially themselves invariants and might help prove the property. An example prompt is given in the bottom left of Figure~\ref{fig:running-example}. Note that this is not the exact prompt we found most useful, but we defer a full discussion of prompts and \emph{prompting strategies} to Sec.~\ref{sec:strategy}. Suppose the oracle returns two potential predicates, both of which should hold at the beginning of the while loop (line 3): \texttt{x\%2==0} and \texttt{x\%4==1}. The \rulepropose rule allows us to make one of them the current assumption. 

\emph{Case (\texttt{x\%2==0}):} The top row illustrates what happens when we transition to $\tuple{\program, \set{\propertyAlt\is\tuple{\texttt{x\%2==0}, 3}}, [\property]}$. While \propertyAlt is indeed an invariant, it does not help to prove the assertion, and \verifier still returns \unknown. This means that the \rulerepairone rule is applicable, which invokes the oracle \prompterRepair to strengthen \propertyAlt. Suppose in this case, the oracle suggests the predicate $\propertyAlt'\is\texttt{x\%4==0}$, which clearly implies the original property \texttt{x!\!=30}. Now, suppose $\verifier(\program, \set{\propertyAlt'}, \property)$ returns \true. We can apply the \ruledecide rule and transition to \tuple{\program, \varnothing, [\property, \propertyAlt']}, making \propertyAlt' the current proof goal. Proving \propertyAlt' is arguably easier because \texttt{x\%4==0} is \emph{inductive} (i.e., if it holds in one iteration, then it will hold in the next iteration), making conventional automated reasoning techniques such as k-induction applicable.
Now, if $\verifier(\program, \varnothing, \propertyAlt') = \true$, we can apply \rulesuccessone and transition to the \success state, thus completing the proof.

\revise{We discuss the case \texttt{x\%4==1} in \extref{App.~\ref{app:subsec:soundness}}.}
%
$\blacksquare$
\end{example}

\section{Instantiating the \sys calculus}
\label{sec:strategy}


In this section, we present strategies for instantiating \sys as an automated procedure. While we showed that the \sys calculus is sound, there are no guarantees that it terminates. Here, we will discuss two sources of non-termination. 

The first one corresponds to an unbounded sequence of suggestions for new sub-goals. Concretely, when trying to prove a particular proof goal \property, we could get stuck if $\verifier(\program, \set{\propertyAlt}, \property) = \unknown$ or $\verifier(\program, \varnothing, \propertyAlt)=\false$ for each proposed assumption \propertyAlt. This could occur as a result of limitations in either the LLM or the verifier. One way to avoid this type of non-termination is by putting an upper bound on the number of proposed assumptions for each proof goal. That is, for any proof goal \property, we require that $\verifier(\program, \set{\propertyAlt}, \property)$ is invoked for at most \numproposals different values of \propertyAlt.

The second source of non-termination is the potentially unbounded depth of the trail $\trail$. Concretely, it is possible to construct an infinite sequence of \rulepropose and \ruledecide transitions, where: (i) the verifier returns \unknown on the current proof goal; (ii) the oracle proposes an assumption that is not invariant but implies the current proof goal; (iii) the verifier proves the implication; (iv) the assumption becomes the new proof goal; and (v) this process repeats. This can be avoided by adding a side condition to the rules requiring that properties proposed by oracles ($\propertyAlt=\tuple{\predicateAlt,\linenumber'}$) must contain a smaller program line number than the one in the current proof goal ($\property=\tuple{\predicate,\linenumber}$), that is, 
\renewcommand{\theequation}{Condition \arabic{equation}}
\begin{equation}
\tuple{\predicateAlt, \linenumber'} \in \prompter_{*}(\program, \tuple{\predicate,\linenumber}, \ldots) \Rightarrow \linenumber' < \linenumber
\label{cond:lowerline}
\end{equation}
Based on these strategies, a terminating (by Thm.~\ref{thm:termination} at the end of this section) and sound (by Thm.~\ref{thm:soundness}) algorithm for checking whether a property \property is an invariant on a program \program is presented in Alg.~\ref{alg:sys}. 
Alg.~\ref{alg:sys} is a recursive procedure $\sysFunc$. It takes a program $\program$ and a property $\property$ as inputs.
If $\sysFunc$ returns \success, then the property is an invariant. If $\sysFunc$ returns \fail, then the property is not an invariant. The function can also return \unknown if the analysis is inconclusive.
At a high level, Alg.~\ref{alg:sys} searches for a potential subgoal \propertyAlt that implies the current goal \property (lines~\ref{algo:gen:start}--\ref{algo:gen:end}). If such a \propertyAlt is identified in line~\ref{algo:imply:check}, we recurse to prove $\propertyAlt$ (line ~\ref{algo:recurse}). The \texttt{while} loop starting at line~\ref{algo:while:start} ensures that at most $k$ attempts can be utilized to generate a new subgoal for $p$.
The comments in Alg.~\ref{alg:sys} indicate which lines correspond to specific proof rules.
The algorithm is sound as it only applies the rules of the calculus.
A full description of Alg~\ref{alg:sys}, including a proof of termination can be found in \extref{App.~\ref{app:sec:strategy}}.

\begin{algorithm}[t]
\small
\begin{algorithmic}[1]
\State {\bfseries Input:} A program \program, a property \property.
\State {\bfseries Output:} \success only if $\invariant(\program, \property)$; \fail only if $\neg\invariant(\program, \property)$; and \unknown if inconclusive. 
\State {\bfseries Parameters:} Verifier \verifier, oracles \prompterPropose and \prompterRepair (which satisfy  ~\ref{cond:lowerline}), number of proposals \numproposals
\Function{$\sysFunc$}{$\program, \property$}
\State {$d \mapsto \verifier(\program, \varnothing, \property)$}\label{algo:init:check}
\If {$d \is \false$} {\return \fail} \Comment{\rulefail}
\ElsIf {$d \is \true$} {\return \success} \Comment{\rulesuccessone}
\Else 
\State {$i, Q \mapsto 0, \prompterPropose(\program, \property)$} \label{algo:gen:start}
\While {$i < \numproposals \land |Q|>0$}  \label{algo:while:start}
\State {$i \mapsto i + 1$}
\State {$\propertyAlt \mapsto \pop(Q)$}
\State {$e \mapsto \verifier(\program, \set{\propertyAlt}, \property)$}  \label{algo:imply:check}  \Comment{\rulepropose/\rulebacktrack}
\If {$e\is \false$} {\return \fail} \Comment{\rulefail}  \label{algo:imply:check:fail}
\ElsIf {$e \is \true$} 
\State {$f \mapsto \sysFunc(\program, \propertyAlt)$} \Comment{\ruledecide}  \label{algo:recurse}
\If {$f \is \success$} {\return \success} \Comment{\rulesuccessone} \label{algo:recurse:output:start} 
\ElsIf {$\stable(\program,\propertyAlt) \land (\verifier(\program, \set{\neg\propertyAlt}, \property) \is \true)$} 
{\return \success} \Comment{\rulesuccesstwo}  \label{algo:recurse:output:flip}
\ElsIf {$f \is \fail$}  \label{algo:recurse:output:fail}
{ $Q\mapsto \join(Q, \prompterRepair(\program, \property, \propertyAlt, \false))$}  \Comment{\rulerepairtwo}  \label{algo:recurse:output:repair}
\Else { \textbf{continue}}
\EndIf\label{algo:recurse:output:end}
\Else { $Q\mapsto \join(Q, \prompterRepair(\program, \property, \propertyAlt, \unknown))$} \Comment{\rulerepairone} \label{algo:gen:end}
\EndIf
\EndWhile
\State {\return \unknown}
\EndIf
\EndFunction
\end{algorithmic}
\caption{The \sys procedure \label{alg:sys}}
\end{algorithm}

\begin{theorem}[Termination]
Given a program \program, and a property \property on the program, Alg.~\ref{alg:sys} terminates with either \success, \fail, or \unknown.
\label{thm:termination}
\end{theorem}




Note that, Alg.~\ref{alg:sys} is one of many possible instantiations of the \sys calculus. One can develop 
alternative strategies to apply the rules, e.g., by changing the frequency of the repair rules and the \rulepropose rules to balance the cost of LLM oracle calls. We evaluate one of the alternatives in \extref{App.~\ref{app:more-exp}}.
\section{Experiments}\label{sec:exp}

We have presented the \sys calculus and described a sound and terminating algorithm based on \sys. In this section, we investigate the following questions:
\begin{itemize}[leftmargin=*]
    \item 
    Can we develop a practical automated verification procedure based on  Alg~\ref{alg:sys}? [Yes]
    \item Is \sys competitive with existing end-to-end learning-based verification approaches? [Yes]
    \item Can \sys already prove hard benchmarks that are beyond the reach of state-of-the-art conventional program verifiers? [In several cases]
\end{itemize}

\subsection{Building an LLM-based program verifier}\label{subsec:implementation}

We report on several practical considerations when building a prototype of Alg.~\ref{alg:sys}. There are two types of external calls that Alg.~\ref{alg:sys} depends on. The first type is calls to \verifier. We use off-the-shelf verifiers in our framework that are extensively tested by the community (described in later paragraphs), so we have some expectations about their reliability and performance. On the other hand, the second type of calls, calls to LLM oracles, introduces more uncertainty, as LLMs are newer and are treated as black boxes.  In our framework,
the oracles \prompterPropose and \prompterRepair automatically prompt a GPT-family model through the OpenAI API and parse the output. We found that while GPT has great potential for generating sensible loop invariants, it still has practical limitations. We report several tactics that we found useful in practice. 
\begin{itemize}[leftmargin=*]
\item \textbf{Formatting the output}:
While investigating whether chain-of-thought (CoT)~\citep{wei2022chain} reasoning is useful when seeking new properties for \program and \property, We found that GPT's outputs, even when containing useful information, were verbose and often contained irrelevant or incorrect statements, making it difficult to extract invariants. 
To address this, we use in-context learning to encourage the LLM to format the output in a specific way. For example, adding \texttt{Your output should be "assert(...);// Line number"} to the prompt is sufficient for GPT to consistently generate outputs of exactly this format, without providing verbose explanations. 
\item \textbf{Inserting markers in the program}: We found that GPT is not good at counting program lines. Often, the generated predicate is otherwise correct, but the line number is slightly off. Unfortunately, an invariant at a wrong position is of no use to the verifier. To mitigate this challenge, we insert placeholder lines of the form \texttt{"// Line A"}, \texttt{"// Line B"} and prompt GPT to generate invariants of the form \texttt{assert(...);// Line name} (for those specific locations). As a simple practical heuristic, we insert placeholders right before loops and at the beginning of loops. 
\item \textbf{Ordering the proposal}: 
The output of an oracle call is non-deterministic for a given prompt, depending on the hyper-parameters of the call. Moreover, the oracles produce a set of properties and we need good heuristics to choose the order of trying those properties. A heuristic we found useful is to prompt GPT multiple times and order the properties by the frequency with which they are proposed (breaking ties by preferring shorter expressions). Moreover, instead of relying on string matching, we treat two proposals the same if their abstract syntax trees are equivalent. 
\end{itemize}

The exact prompts are described in \extref{App.~\ref{sec:prompts}}. For \verifier, we consider two state-of-the-art formal tools for C program verification, \esbmc~\citep{esbmc} and \uautomizer~\citep{uautomizer}. The former is based on k-induction and the latter is based on predicate abstraction. \esbmc and \uautomizer were the top two performing non-portfolio solvers in the reachability track of the most recent edition of the software verification competition (SV-COMP~\citep{SVCOMP}). Furthermore, \uautomizer was the overall winner.
By default, we impose a 30-second time limit for each invocation of the verifier. That is, if the verifier does not terminate within 30 seconds, then the verifier returns \unknown.%
\footnote{The source code and the benchmarks are publicly available at \url{https://github.com/wu-haoze/Lemur-program-verification}.}

\subsection{Loop invariant generation benchmarks}

A prominent approach in learning-based end-to-end program verification is Code2Inv~\citep{code2inv}, which uses reinforcement learning to train an invariant synthesizer to propose loop invariants. At a high level, to infer a loop invariant for a given program, Code2Inv learns a generative neural network whose goal is to repeatedly propose candidate invariants until an automated reasoning tool determines that the candidate is correct. Throughout this process, the generator is fine-tuned using the feedback from the automated reasoning tool. 
In this section, we study how \sys compares with this approach on the same benchmark set considered in the original Code2Inv work. The benchmark set contains 133 benchmarks, each containing a C program and a property expressed as an assert statement in the program. Each program contains a single loop and each loop can have nested if-then-else blocks (without nested loops). 
The assertion to check is always after the loop. 
Code2Inv is designed to generate invariants at the beginning of the while loop. We prompt the oracles to generate invariants at the same location.


We use the k-induction-based verifier, \esbmc, to check the implication (line 13 in Alg.~\ref{alg:sys}) which aligns with the verification procedure used in Code2Inv.
We report the number of solved instances as well as the number of failed suggestions (either the suggestion itself cannot be verified or \esbmc times out on the implication check). As a point of comparison, we report the statistics from the original Code2Inv approach.
Code2Inv was given a one-hour timeout. In addition, we also report \esbmc's performance on this set of benchmarks.  The result is shown in Table~\ref{tab:cov2inv} \revise{(the Time column shows the average time in seconds to solve a benchmark).}

\begin{wrapfigure}{r}{0.3\textwidth} 
\vspace{-8mm}
  \begin{center}
    \includegraphics[width=0.3\textwidth]{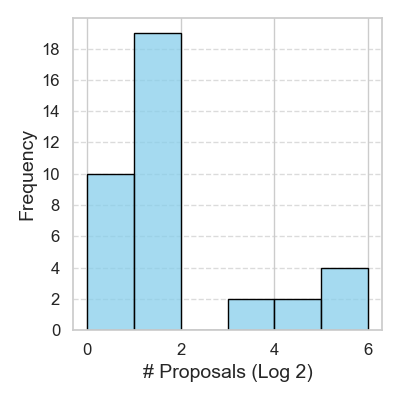}
    \vspace{-0.8cm}
    \caption{Nb. of proposals for \systable to solve a \revise{Code2Inv}  benchmark.}
    \vspace{-0.5cm}
    \label{fig:histogram-esbmc}
  \end{center}
\end{wrapfigure}
With a 10-minute timeout, \esbmc alone can solve 68 problems. On the other hand, \systable can solve 107 problems within the same time limit.
Surprisingly, this approach solves more instances than Code2Inv, which is specifically designed for invariant synthesis tasks. 
For problems unsolved by \esbmc but solved by \systable, a histogram of the values of $Log_2$ of the number of proposals is shown in Fig.~\ref{fig:histogram-esbmc}. While in most cases, Alg.~\ref{alg:sys} can produce the correct proposals within 4 attempts, 
there are benchmarks for which \systable requires many iterations to find the desired loop invariant, e.g.,
one of the benchmarks took 90 proposals.
\revise{In addition, we experimented with the GPT-3.5 turbo LLM model, denoted \sysalt,  as shown in Table~\ref{tab:cov2inv}. Note that \sysalt solved four fewer benchmarks and required more time and more calls to the GPT-3.5 turbo oracle.}

\begin{table}[t]
  \begin{subtable}{0.5\linewidth}
  \setlength\tabcolsep{2pt}
    \centering
    \begin{tabular}{cccc}
    \toprule
        Configurations &  Solved & Time \revise{(s)}  & \# proposal \\
            \midrule
     Code2Inv & 92 & -- & - \\
        \esbmc & 68 & 0.34 & 0 \\
        \revise{\sysalt} & \revise{103} & \revise{35.6} & \revise{8.6} \\
        \systable & 107 & 32.9 & 4.7 \\
         \bottomrule
    \end{tabular}
    \caption{
    The Code2Inv benchmarks.
    }
    \label{tab:cov2inv}
  \end{subtable}%
  \begin{subtable}{0.5\linewidth}
\setlength\tabcolsep{2pt}
    \centering
    \begin{tabular}{cccc}
    \toprule
        Configurations & Solved  & Time (\revise{s}) & \# proposals \\
            \midrule
     \uautomizer & 1 & 824.3 & 0 \\
        \esbmc & 1 & 675.7 & 0 \\
        \revise{\sysalt} & \revise{14} & \revise{162.2} & \revise{8.5} \\
        \systable & 25 & 234.5 & 7.2 \\
         \bottomrule
    \end{tabular}
    \caption{
    The  47 SV-COMP benchmarks.
    }
    \label{tab:svcomp}
  \end{subtable}
  \vspace{-4mm}
  \caption{Solved instances by  \esbmc,  \sys, and Code2Inv (\ref{tab:cov2inv}) or \uautomizer (\ref{tab:svcomp}) on two benchmark sets. We also report the average time and number of proposals for solved instances.}
  \vspace{-0.5cm}
  \label{tab:main}
\end{table}

\subsection{Solving hard SV-COMP benchmarks}

Next, we study \sys's ability to solve hard benchmarks from the 2003 edition of SV-COMP~\citep{SVCOMP}. 
We focus on benchmarks with less than 150 tokens (after removing comments and unnecessary headers, and after applying clang-formatting).
We select 47 benchmarks that \esbmc and \uautomizer are unable to solve within 10 minutes. 
The property is expected to hold in all benchmarks. 
We use a 15-minute per-instance timeout.

The results are shown in Table~\ref{tab:svcomp}. Impressively, with the guidance of the proof goals suggested by the LLM, \systable is able to solve 25 of the 47 SV-COMP benchmarks. While \esbmc and \uautomizer can each solve only 1 benchmark. 
Upon closer examination, 
8 of the solved instances contain two loops and
5 contain three or more loops. This suggests that \sys is capable of (i) handling programs with more than one loop; and (ii) boosting the performance of state-of-the-art conventional C program verifiers.

\begin{wrapfigure}{r}{0.3\textwidth} 
\vspace{-8mm}
  \begin{center}
    \includegraphics[width=0.3\textwidth]{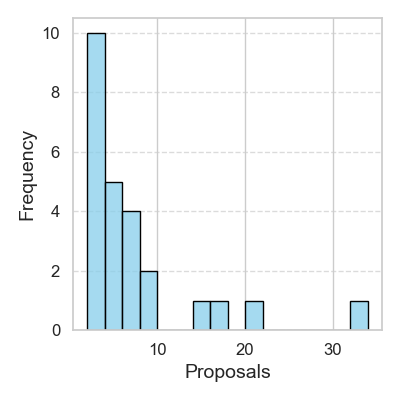}
    \vspace{-0.8cm}
    \caption{Nb. of proposals for \systable to solve an \revise{SV-COMP} benchmark.}
    \label{fig:histogram-svcomp}
    \vspace{-0.7cm}
  \end{center}
\end{wrapfigure}
The average number of proposals before solving a problem is higher compared to the Code2Inv benchmarks (7.2 vs. 4.7). Fig.~\ref{fig:histogram-svcomp} sheds more light on the behavior of \systable.~
We observe that 12 of the 26 solved instances require at least 5 proposals in total. 

We found that LLM oracles can produce surprisingly insightful loop invariants which are difficult for conventional formal methods to synthesize. While predicate-abstraction-based techniques typically generate predicates that involve only the operators and values in the program and follow a particular template,  LLMs are not constrained by these limitations. For example, for the program in Fig.~\ref{fig:running-example}, GPT-4 can consistently generate \texttt{x\%4==0} as the loop invariant although the modulo operator is not present in the program. In another example,  the LLM correctly understands the range of an \texttt{unsigned char} and suggests variable bounds as the assumption, which ends up being the key to proving the property.  This example is shown in \extref{App.~\ref{subsec:prompt-propose}}. There are also several cases where the LLM generates disjunctive invariants that precisely characterize the behavior of the loops.

\revise{Finally, we observe that \systable significantly outperforms \sysalt across all metrics. This suggests that the choice of oracle is also crucial for performance.} Additional experiments are presented in \extref{App.~\ref{app:more-exp})}. They include running the baseline solvers with a 12 hour timeout, using a configuration in which repair rules are not employed in Alg.~\ref{alg:sys}, and running \sys multiple times to account for the stochasticity of the oracles.


\section{Discussion of limitations and extensions}
\vspace{-2mm}
In this work, we propose a novel framework, \sys, which combines automated reasoning and LLMs. To the best of our knowledge,  \sys is the first framework to provide a theoretical foundation for such an integration via a formal calculus.
We also implemented \sys as a fully automated framework and demonstrated its efficiency on standard benchmark sets.
We conclude by discussing the current limitations of \sys,  which also point to future research directions.

As mentioned above, the practical performance of \sys depends on two types of external calls: the verifiers and the LLMs. Any improvements in these tools should translate into \sys improvements. Currently, modern verifiers are capable of handling only relatively small programs (see SV-COMP'23~\citep{SVCOMP}). Interestingly, even when provided with a strong invariant, they sometimes cannot solve the verification problem. One research direction that we envision is to customize \sys to a particular back-end verifier to obtain better performance and solve larger programs.

\revise{We also note that \sys primarily focuses on imperative languages. Extending it to functional languages is a direction for future research.}

While our experience with LLMs was largely positive (see Section~\ref{subsec:implementation} for a discussion of limitations that have already been at least partially addressed), 
there are more interesting challenges to tackle.
First, current LLMs have token limits, and many practical programs exceed those limits. 
Second, it is sometimes challenging for LLMs to generate complex logical formulas such as those with nested if-then-else expressions. 
We believe that to overcome this limitation, we need to (i) develop a prompting language for LLM invariant generation, and (ii) fine-tune LLMs for invariant generation tasks using this language.
\revise{Third, reasoning about programs with multiple loops remains challenging for LLMs. We believe fine-tuning could help address this challenge.} 
\revise{Fourth, we observed that the performance of \sys may vary depending on the LLM oracle. For example, our experiments demonstrate that GPT-4 is superior to GPT-3.5 turbo on tested benchmarks.
}
%
Finally, due to the limitations of LLMs and automated reasoners, our framework does not yet offer a significant boost for complex properties of real-world C libraries. 
%
%
However, a modular approach, where large parts of the program are abstracted and summarized in the form of pre- and post-conditions, can benefit from frameworks like \sys.

\newpage
\bibliographystyle{iclr2024_conference}
\bibliography{bibli}

\ifextended
\newpage
\appendix



\section{Background}
\label{app:background}
 \paragraph{Automated reasoning tools.}
We  overview several popular techniques that are used by modern program verification solvers, like  \cbmc~\citep{cbmc}, \esbmc~\citep{esbmc}, and \uautomizer~\citep{uautomizer}.

The Bounded Model Checking (BMC) approach is an iterative technique that verifies the program for each unwind bound up to a maximum value, $m$, e.g.
it can perform  $m$ loop unwinding. It either finds a counterexample within $m$ steps or returns \unknown.  This approach usually employs SMT solvers to find counterexamples very efficiently  at each step. However,
for non-trivial systems unrolling can be expensive and memory-consuming.
Moreover, vanilla BMC can only check finite reachability, e.g. it 
cannot prove loop invariants, for example.

Another popular technique is the k-induction method, which allows BMC-style methods to prove properties like loop invariants. This approach is also iterative. First, we check whether a property holds for $k$ steps from a valid state. If it does not, we find a counterexample. Otherwise, we check an inductive hypothesis that if the property holds for $k$ steps from a state then the property holds for the $k+1$th step. If it does, the property holds for any number of steps. If not, k-induction either increases $k$ or returns \unknown.
As in the case of BMC, unrolling can be computationally expensive. Moreover, k-induction is not complete; there are properties that are not k-inductive for any $k$.

The last approach we consider is  abstract interpretation verification methods.
Such methods create abstract representations of program states and variables. These abstract representations are simplified versions of the actual program states, focusing on relevant aspects while ignoring irrelevant details. The choice of abstraction depends on the specific properties to verify. Moreover, if a property holds for an abstract program, then it holds for the original program. The reverse is not true. Hence, if a property does not hold for an abstract program, we need to refine the abstract representation to exclude the counterexample. The main challenge here is how to come up with a good abstraction and how to design a refinement procedure.

\paragraph{Large Language Models.}
Large Language Models belong to a class of artificial intelligence models used in natural language processing tasks. These models are designed to process and generate both human language and structured inputs, such as code in various programming languages. 
Examples of large language models include Generative Pre-trained Transformer models like GPT-3~\citep{gpt3} or GPT-4~\citep{openai2023gpt4},  Bidirectional Encoder Representations from Transformers, BERT~\citep{bert}, and others. 
LLMs are getting increasingly popular as an AI-assistance for code generation tasks, like PaLM~\citep{palm},  GitHub Copilot~\citep{copilot,githubcopilot}, etc.

LLMs are usually trained in two steps. The main phase is training, where these models are exposed to very large corpora of data, usually collected over the internet. The architecture of LLMs is based on transformers and has a very large number of parameters. Therefore, it can capture relations between different parts of the input text to  produce coherent outputs. For some applications, we need to perform fine-tuning to expose the model to application-specific inputs. During inference, when a user provides inputs and a prompt that contains instructions to an LLM, it generates the output with respect to these instructions. 

\section{Definitions}\label{app:sec:def}

\begin{replemma}{lemma:stable}
Consider a program \program, two properties \property, \propertyAlt on \program, and a program $\programAlt\is\asm(\program,\propertyAlt)$.
The property \property is an invariant on $\program$, if 1) \propertyAlt is stable for $\program$; 2) \propertyAlt implies \property with respect to $\program$; and 3) $\neg\propertyAlt$ implies \property with respect to $\program$. More formally:
$\stable(\program,\propertyAlt) \wedge (\imply{\propertyAlt}{\program}{\property}) \wedge (\imply{\neg\propertyAlt}{\program}{\property})\Rightarrow \invariant(\program, \property).$
\end{replemma}
\vspace{-0.6cm}
\begin{proof}
Let $\propertyAlt\is\tuple{\predicate,\linenumber}$.
By the definition of stability, for any execution of \program, either \predicate always evaluates to true at line \linenumber or $\neg\predicate$ always evaluates to true at line \linenumber. In either case, the property \property holds by the definition of implication. Therefore, \property holds for all executions of \program, i.e., $\invariant(\program,\property)$.
\end{proof}

\section{\sys: Integrating LLMs in Program Verification}

\subsection{Soundness of \sys}
\label{app:subsec:soundness}

\begin{lemma}
For any configuration \tuple{\program, \assumptions, \trail} created by a sequence of valid rule applications starting from an initial configuration \tuple{\program, \varnothing, [\property_0]}, \trail is not empty.
\label{lemma:non-empty}
\end{lemma}
\begin{proof}
This can be proven by induction on the length of the sequence. In the base case, \trail is $[\property_0]$. In the inductive case, the length of \trail does not reduce except in the \rulebacktrack rule which requires \trail to have at least 2 elements in the pre-condition. Thus, \trail is not empty.
\end{proof}

\begin{lemma}
Let \tuple{\program, \assumptions, \trail} be a configuration created by a sequence of valid rule applications starting from an initial configuration \tuple{\program, \varnothing, [\property_0]}, and let \property be the last element of \trail. We have \imply{\property}{\program}{\property_0}.
\label{lemma:implication}
\end{lemma}
\begin{proof}
We prove a stronger property, that for \emph{each} element \property in \trail, $\imply{\property}{\program}{\property_0}$. We induct on the length of the sequence. In the base case, \imply{\property_0}{\program}{\property_0} by proposition~\ref{prop:self-imply}. In the inductive case, we proceed by cases. \rulesuccessone, \rulesuccesstwo, \rulefail cannot be applied. In the post conditions of \rulepropose, \rulebacktrack, \rulerepairone, and \rulerepairtwo, \trail either shrinks or remains the same. Therefore, the inductive hypothesis can be directly applied. If \ruledecide rule is to be applied. In the pre-condition, the trail is $\trail\concat\property$, the current assumption is \set{\propertyAlt} and 
$\imply{\propertyAlt}{\program}{\property}$. In the post condition, \trail becomes $\trail\concat\property\concat\propertyAlt$.
By the inductive hypothesis, $\imply{\property}{\program}{\property_0}$. Furthermore, by Proposition~\ref{prop:chain-imply}, $\imply{\propertyAlt}{\program}{\property_0}$.
\end{proof}

\begin{lemma}
Let \tuple{\program, \assumptions, \trail\concat\property\concat\property'} be a configuration created by a sequence of valid rule applications starting from an initial configuration \tuple{\program, \varnothing, [\property_0]}, we have \imply{\property'}{\program}{\property}.
\label{lemma:implication2}
\end{lemma}
\begin{proof}
This can be proven by induction on the length of the sequence.
\end{proof}

\begin{reptheorem}{thm:soundness}
Given a property $\property_0$ and a program \program, if \success is reached by a sequence of valid rule applications starting from \tuple{\program, \varnothing, [\property_0]}, then $\property_0$ is an invariant on \program.
\end{reptheorem}
\begin{proof}
We can transition into \success by either the \rulesuccessone rule or the \rulesuccesstwo rule. In the pre-condition of \rulesuccessone, the trail is of the form $\trail\concat\property$, and the verifier \verifier proves that \invariant(\program, \property). By Lemma~\ref{lemma:implication}, \imply{\property}{\program}{\property_0}. Further by Proposition~\ref{prop:imply}, we have $\invariant(\program, \property_0)$. 
On the other hand,  in the pre-condition of \rulesuccesstwo, the trail is of the form $\trail\concat\property\concat\property'$. By Lemma~\ref{lemma:implication2}, $\imply{\property'}{\program}{\property}$. In addition, $\property'$ is stable and $\imply{\neg\property'}{\program}{\property}$. Therefore, by Lemma~\ref{lemma:stable}, \property is an invariant of \program. Since we also know from Lemma~\ref{lemma:implication} that $\imply{\property}{\program}{\property_0}$, it follows from Proposition~\ref{prop:imply} that $\property_0$ is an invariant of \program.
\end{proof}

\begin{reptheorem}{thm:soundness2}
Given a property $\property_0$ and a program \program, if \fail is reached by a sequence of valid rule applications starting from \tuple{\program, \varnothing, [\property_0]}, then $\property_0$ is not an invariant on \program.
\end{reptheorem}
\begin{proof}
We transition into the \fail state only when the verifier $\verifier(\program, \assumptions, \property_0) = \false$. Even if \assumptions is not empty, $\property_0$ is still not an invariant by Prop.~\ref{prop:imply}. 
\end{proof}

\begin{example} \label{app:subsec:soundness:full example}
To provide more intuition about the proof system and to motivate the design choices when instantiating \sys, we consider again our running example. Figure~\ref{fig:running-example} illustrates how \sys can be used to verify properties in practice. In Figure~\ref{fig:running-example} each frame represents a state of the program. Transitions between states are depicted by arrows, with each arrow marked with the rule applied to execute this transition. 
In this example, our goal is to prove the property \texttt{x!\!=30} in a while loop that keeps adding 4 to an unsigned 32-bit integer variable \texttt{x}. We note that this particular verification task is adapted from a similar one in the SV-COMP competition.\footnote{\url{https://sv-comp.sosy-lab.org/2023/results/results-verified/META_ReachSafety.table.html\#/table?filter=id_any(value(jain_5-2))}}
While seemingly trivial, during the competition, 19 out of the 24 participating tools (including the overall winner of the competition \uautomizer) were not able to solve this benchmark. 

Given such a verification problem, we create an initial configuration \tuple{\program, \varnothing, [\property])} where \program is the given problem and $\property\is\tuple{\texttt{x!\!=30}, 3}$.\footnote{3 is the line number (in the snippet) where the predicate is asserted.} 
Suppose the verifier \verifier is unable to solve this problem and returns \unknown. In this case, we need to generate a new proof goal, so the only rule we could apply is \rulepropose. To do so, we invoke the LLM-based oracle \prompterPropose to obtain a set of new properties that are potentially themselves invariants and might help prove the property. An example prompt is given on the left bottom part. This is not the exact prompt that we found the most effective in practice and we defer the discussion of \emph{prompting strategies} to Sec.~\ref{sec:strategy}. Suppose the oracle returns two potential predicates at the beginning of the while loop: \texttt{x\%2==0} and \texttt{x\%4==1} at line 3. The \rulepropose allows us to make one of them the current assumption. 

\emph{Case (\texttt{x\%2==0}):} The top row illustrates what happens when we transition into $\tuple{\program, \set{\propertyAlt\is\tuple{\texttt{x\%2==0}, 2}}, [\property]}$. While \propertyAlt is indeed an invariant, it does not help to prove the assertion and \verifier would return \unknown. This satisfies the condition to apply the \rulerepairone rule, which would invoke the oracle \prompterRepair to strengthen \propertyAlt. Suppose in this case, the oracle suggests the predicate $\propertyAlt'\is\texttt{x\%4==0}$, which clearly implies the original property \texttt{x!\!=30}. Suppose then $\verifier(\program, \set{\propertyAlt'}, \property)$ returns \true. We could apply the \ruledecide rule and transition to \tuple{\program, \varnothing, [\property, \propertyAlt']}, making \propertyAlt' the current proof goal. Proving \propertyAlt' is arguably easier because \texttt{x\%4==0} is \emph{inductive} (i.e., if it holds in one iteration and then it will hold in the next iteration), making conventional automated reasoning techniques such as k-induction applicable.
Suppose $\verifier(\program, \varnothing, \propertyAlt') = \true$, we could apply \rulesuccessone and transition into the \success state, thus completing the proof.

\emph{Case (\texttt{x\%4==1}):} 
The bottom row illustrates a different chain of rule applications when we picked the property $r\is\tuple{\texttt{x\%4==1}, 2}$ from the first proposal. While $r$ does not hold, it does imply \texttt{x!\!=30}. Suppose this implication is proven by the verifier. We could apply \ruledecide and transition to \tuple{\program, \varnothing, [\property, r]}. Since $r$ is not an invariant, $\verifier(\program, \varnothing, r)$ would be either \unknown or \false. Either way, we could apply \rulebacktrack and try a new assumption proposed by \prompterPropose. In practice, we could either invoke the stochastic \prompterPropose again or pick an un-attempted property (e.g., \tuple{\texttt{x\%2==0}, 2} proposed previously). In the illustration, we invoke \prompterPropose again and obtain the ``correct'' predicate \texttt{x\%4==0}, which would allow us to prove the property in two more steps.
$\blacksquare$
\end{example}

\section{Instantiating the \sys calculus}
\label{app:sec:strategy}

Here, we  describe Alg.~\ref{alg:sys}. First, the algorithm checks whether the current \property can be verified by \verifier or if a counterexample exists (line~\ref{algo:init:check}). If so, it returns either \success or \fail to the upper level of recursion or terminates if $\sysFunc$ is at the top level. If \verifier cannot prove $\property$, i.e. it returns \unknown, $\sysFunc$ enters a new phase of subgoal generation, where LLM oracles are used to propose new or repair existing properties (lines~\ref{algo:gen:start}--\ref{algo:gen:end}). In this phase, we start by calling  \prompterPropose to generate a new subgoal (line~\ref{algo:gen:start}). The \texttt{while} loop at line~\ref{algo:while:start} ensures that at most $k$ attempts can be unitized to generate a new subgoal for $p$. 
In line~\ref{algo:imply:check}, we call \verifier to check whether $\propertyAlt$ implies $\property$. If \verifier returns \false, we know that $\property$ is not an invariant and return \fail (line~\ref{algo:imply:check:fail}). If \verifier returns \unknown, then we need to repair $\propertyAlt$; for example, we might strengthen $\propertyAlt$ and try again to prove implication.
Otherwise, if $\propertyAlt$ \emph{does} imply $\property$, we recurse to prove $\propertyAlt$ (line~\ref{algo:recurse}). The last logical block of $\sysFunc$ in lines~\ref{algo:recurse:output:start}--\ref{algo:recurse:output:end} addresses the output of the recursive call. If we have successfully proved that $\propertyAlt$ is an invariant, we return \success. Otherwise, if $\propertyAlt$ is stable (see Definition~\ref{def:stable}), we can check whether $\neg\propertyAlt$ implies $\property$ (line~\ref{algo:recurse:output:flip}). If so, by Lemma~\ref{lemma:stable}, we can conclude that $\property$ is an invariant and also return \success. If we prove that $\propertyAlt$ is \false, we can repair $\propertyAlt$ by informing an LLM oracle that the property does not hold (line~\ref{algo:recurse:output:repair}).
%
Finally, if $f$ is \unknown  then we continue to the next iteration of the \texttt{while}  loop and consider the next proposed sub-goal.

Second, we present a proof of Theorem~\ref{thm:termination}.
\begin{reptheorem}{thm:termination}
Given a program \program, and a property \property on the program, Alg.~\ref{alg:sys} terminates with either \success, \fail, or \unknown.
\end{reptheorem}
\begin{proof}
Suppose $\property\is\tuple{\predicate,\linenumber}$. We prove with a decreasing argument on $\linenumber$. When $\linenumber = 0$, the algorithm terminates without entering the while loop, because \prompterPropose satisfies ~\ref{cond:lowerline} and $\prompterPropose(\program,\property) = \varnothing$.
In the recursive case, the while loop is executed for at most \numproposals iterations. In each iteration, we show that for the second input to $\sysFunc$ (Line 16), $\propertyAlt\is\tuple{\predicateAlt, \linenumber'}$, we have $\linenumber' < \linenumber$. This is true because \propertyAlt is generated either by $\prompterPropose(\program, \property)$ or $\prompterRepair(\program, \property, \ldots)$, both satisfying ~\ref{cond:lowerline}.
\end{proof}

\section{Related work}
\begin{table}[h!]
\centering
\scriptsize
 \begin{tabular}{|c||cccc|} 
 \hline
 & \revise{ML model} &\revise{Learning scheme} & \revise{Restrictions on}&  \revise{Framework purpose}  \\
 \hline\hline
\cite{0001NMR16} & \revise{Decision trees} & \revise{Dynamic learning via examples} &   \revise{invariant  structure} & \revise{invariant generation}  \\
 Code2Inv\citep{code2inv} & \revise{NNs}  & \revise{RL-based training}  &  \revise{program structure}  & \revise{verification (single invariant)}\\
 \cite{pmlr-v202-pei23a} & \revise{LLMs} & \revise{LLMs' fine-tuning} &  \revise{18 invariant types} & \revise{invariant generation}  \\  
 \sys (this work) & \revise{LLMs}  & \revise{--}  & \revise{--} & \revise{verification (a chain of invariants)}\\ 
 \hline
 \end{tabular}
 \caption {\revise{A comparison of representative learning-based invariant generation techniques.}}\label{table:rel}
\end{table}



There has been a lot of interest in using LLMs to augment formal reasoning.
\cite{charalambous2023new} proposed a novel framework, ESBMC-AI, that integrated LLMs reasoning and formal verification. They also applied their framework to the analysis of C programs focusing on  memory safety properties.
The main idea is to use LLMs as a code repair generator that can fix faulty code using a carefully designed prompt, a program, and a counterexample provided by a bounded model checker. However, ESBMC-AI assumes that program rewriting done by an LLM is valid, i.e.  syntactically and semantically correct. The latter is challenging to prove in an automatic manner as it requires program equivalence checking. Our framework does not use LLMs to modify code and treat the outputs of the LLM as suggestions until we prove that they are correct. Another example of an automated framework is Baldur~\citep{Baldur}, which uses an LLM, Minerva \citep{minerva}, to generate theorem proofs that are checked by Isabelle theorem prover. They also proposed a proof repair procedure.
In contrast, our interactive decision procedure relies on an automated reasoner to generate proofs and only uses LLMs generated program properties. 

\revise{There are a number of learning approaches to automatically generate invariants~\citep{code2inv,0001NMR16,0001GHALN13}. \cite{0001GHALN13} proposed a data-driven iterative approach to derive algebraic equation invariants from data generated through concrete executions of the program. 
\cite{0001NMR16} proposed a decision tree-based approach to learn invariants from examples; however, the space of possible invariants is limited to a logical combination of binary linear constraints.} The most related work to our framework is Code2Inv \citep{code2inv}, which proposed learning program invariants using machine learning techniques and employed automatic reasoning to verify the programs. The main principle of partitioning responsibilities between automated reasoners and LLMs is similar to our framework. However, we provide a formalization for such interactive procedures with formal calculus and a strategy to use it in practice. Our procedure is more general as it allows the generation of sequences of logically related properties, and we demonstrate that it is more efficient in practice. Finally, recent work by \cite{pmlr-v202-pei23a} investigates the potential of invariant generation for Java programming language. While this framework does not incorporate automated reasoning components, it shows the potential of LLMs to uncover program properties. 
\revise{Table~\ref{table:rel}  presents a summary of learning-based approaches.}

\section{Prompting the GPT}\label{sec:prompts}

In this section, we describe how we automatically constructed the prompts in \prompterPropose and \prompterRepair, and show examples of the prompts and the GPT outputs. We provided in the supplementary materials the execution traces of \sys on solved benchmarks used in our experiments. 

\subsection{Proposing new properties}\label{subsec:prompt-propose}

Given a program $\program$ and a property represented as a C assert statement in $\program$, we inserted the placeholder lines ``\texttt{// Line A}'', ``\texttt{// Line B}''... to dedicated program lines as described in ~\ref{subsec:implementation}. Our prompt has the following structure:

\begin{verbatim}
---------
[P1]
Print [P2] as valid C assertions at line[P3] [P4] that 
help prove the assertion. Use '&&' or '||' if necessary. 
Don't explain. Your answer should be 'assert(...); // line [P5]'
---------
\end{verbatim}

\texttt{P1} is the program with the placeholder lines and the C assertion; \texttt{P2} is either \texttt{facts} or \texttt{loop invariants} depending on whether all the placeholder lines are in loops; \texttt{P3} is either empty or \texttt{s} depending on the number of placeholder lines; \texttt{P4} are the placeholder line names (e.g., \texttt{A, B}); and \texttt{P5} is either \texttt{A}, if there is only one placeholder line, or \texttt{name}, if there are multiple.

Below is an example extracted from the execution of our prototype on an SV-COMP benchmark:
\begin{verbatim}
int main() {
unsigned char n = (unsigned char) rand();
if (n == 0) {
return 0;
}
unsigned char v = 0;
unsigned int s = 0;
unsigned int i = 0;
while (i < n) {
// Line A
v = (unsigned char) rand();
s += v;
++i;
}
assert(s >= v);
return 1;
}
Print loop invariants as valid C assertions at line A that
help prove the assertion. Use '&&' or '||' if necessary. 
Don't explain. Your answer should be 'assert(...); // line A'
\end{verbatim}

Below are GPT-4's responses on multiple attempts as well as the final set of extracted properties our prototype proceeds to prove.
\begin{verbatim}
GPT output 1 with penalty 1.5:
assert(i <= n); // line A
assert(s >= i); // line A
GPT output 2 with penalty 1.5:
assert(i <= n); // line A
assert(s >= 0 && v >= 0); // line A
GPT output 3 with penalty 1.5:
assert(i <= n); // line A
assert(s >= 0 && s <= i*UCHAR_MAX); // line A
GPT output 4 with penalty 1.5:
assert(i <= n); // line A
assert(s >= 0 && s <= i*UCHAR_MAX); // line A
GPT output 1 with penalty 2:
assert(i <= n); // line A
assert(s >= 0 && s <= UCHAR_MAX * i); // line A
GPT output 2 with penalty 2:
assert(i <= n); // line A
assert(s >= i*v && s <= (i*255)); // line A
GPT output 3 with penalty 2:
assert(i <= n); // line A
assert(s >= 0 && s <= i*UCHAR_MAX); // line A
GPT output 4 with penalty 2:
assert(i <= n); // line A
assert(s >= 0 && s <= UCHAR_MAX * i); // line A

Found 6 potential sub-goals
Goal 1: i <= n after line 8
Goal 2: s >= 0 && s <= i*0xff after line 8
Goal 3: s >= 0 && s <= 0xff * i after line 8
Goal 4: s >= i after line 8
Goal 5: s >= 0 && v >= 0 after line 8
Goal 6: s >= i*v && s <= (i*255) after line 8
\end{verbatim}

\subsection{Repairing a previous proposal}\label{subsec:prompt-repair}

The oracle \prompterRepair uses a similar prompting template:
\begin{verbatim}
---------
[P1]
Print [P2] as valid C assertions at line A that 
help prove the assertion. Your previous answer [P3]
is [P4]. Use '&&' or '||' if necessary. Don't explain. 
Your answer should be 'assert(...);'
---------
\end{verbatim}

\texttt{P1} is the program with the placeholder line (corresponding to the assumption $\property \is \tuple{\predicate, \linenumber}$ to repair) and the C assertion; \texttt{P2} is either \texttt{facts} or \texttt{loop invariants} depending on whether \property is in a loop; \texttt{P3} is \property; \texttt{P4} is either \texttt{incorrect} or \texttt{too weak}.

Below is an example extracted from the execution of our prototype on another SV-COMP benchmark:
\begin{verbatim}
int main() {
unsigned int n = (unsigned int) rand();
unsigned int x = n, y = 0, z;
while (x > 0) {
x--;
y++;
}
z = y;
// Line A
while (z > 0) {
x++;
z--;
}
assert(y + z == n);
return 1;
}
Print facts as valid C assertions at line A that help 
prove the assertion. Your previous answer 'x + z == n' 
is too weak. Use '&&' or '||' if necessary. Don't explain. 
Your answer should simply be 'assert(...);'
\end{verbatim}

The GPT outputs on different prompting attempts with different penalties and the extracted properties:
\begin{verbatim}
GPT output 1 with penalty 1.5:
assert(x + y == n); // Line A
assert(x + z == y); // Line B
GPT output 2 with penalty 1.5:
assert(x + y == n); // Line A
assert(x + z == n); // Line B
GPT output 1 with penalty 2:
assert(x + y == n); // Line A
assert(x + z == n); // Line B
GPT output 2 with penalty 2:
assert(x + y == n); // Line A
assert(z + x == n); // Line B

Found 4 potential adapted sub-goals
Goal 1: x + y == n after line 7
Goal 2: x + z == n after line 7
Goal 3: x + z == y after line 7
Goal 4: z + x == n after line 7
\end{verbatim}

\section{Additional experimental results.}
\label{app:more-exp}

In this section, we report additional results on the SV-COMP benchmarks. First, we run the baseline solvers with a 12-hour timeout instead of just 15 minutes; second, we consider an additional \sys configuration, \sysnorefine, which is the same as \systable except that we do not apply the repair rules (\rulerepairone, \rulerepairtwo); finally, we run each \sys configurations three times (A, B, and C), to account for the stochasticity of the oracles. Results are shown in Tab.~\ref{tab:additional}.

We observe that 1) running the baseline solvers with a longer runtime does not result in significantly more solutions; 2) the repair rules contribute to the solving of more instances; 3) while an un-negligible variance exist across different runs of the same \sys configuration, the conclusions drawn in Sec.~\ref{sec:exp}, that \sys can boost the performance of traditional program verifiers, and that the choice of the underlying oracle can significantly affect performance, remain valid.

\begin{table}[h]
\setlength\tabcolsep{4pt}
\centering
\begin{tabular}{cccccc}
\toprule
Configurations & Time limit &  Run  &Solved  & Time (\revise{s}) & \# proposals \\
\midrule
\uautomizer & 15 mins & - & 1 & 824.3 & 0 \\
\esbmc & 15 mins & - & 1 & 675.7 & 0 \\
\uautomizer & 12 hrs & - & 2 & 1304.0 & 0 \\
\esbmc & 12 hrs & - & 4 & 1637.7 & 0 \\
\midrule
\revise{\sysalt} & 15 mins & A & \revise{14} & \revise{162.2} & \revise{8.5} \\
\revise{\sysalt} & 15 mins & B & \revise{15} & \revise{179.3} & \revise{5.4} \\
\revise{\sysalt} & 15 mins & C & \revise{15} & \revise{103.3} & \revise{4.8} \\
\midrule
\sysnorefine & 15 mins & A & 19 & 139.0 & 4.3 \\
\sysnorefine & 15 mins & B & 20 & 146.1 & 4.8 \\
\sysnorefine & 15 mins & C & 20 & 145.1 & 5.4 \\
\midrule
\systable & 15 mins & A & 25 & 234.5 & 7.2 \\
\systable & 15 mins & B & 24 & 170.2 & 6.1 \\
\systable & 15 mins & C & 24 & 201.9 & 6.3 \\
\bottomrule
\end{tabular}
\caption{
    All experiments conducted on the  47 SV-COMP benchmarks.
    }
\label{tab:additional}
\end{table}
\fi

\end{document}